\documentclass{scrartcl}
\title{Sparse Power Factorization: Balancing peakiness and sample complexity \thanks{
		The results of this paper have been presented in part at the 12th International Conference on Sampling Theory and Applications, July 3-7, 2017, Tallinn, Estonia \cite{geppert2017refined} and in the IEEE Statistical Signal Processing Workshop 2018, June 10-13, Freiburg, Germany \cite{stoeger2018refined}.}}
\author{Jakob Geppert\thanks{Institute for Numerical and Applied Mathematics, Universit\"at G\"ottingen, 37083 G\"ottingen, Germany}, Felix Krahmer\thanks{Department of Mathematics, Technische Universit\"at M\"unchen, 85748 Garching/Munich, Germany}, Dominik St\"oger\footnotemark[3]}

\usepackage[utf8]{inputenc}
\usepackage[T1]{fontenc}
\usepackage[british]{babel}

\usepackage{microtype} %Typographische Optimierung
\usepackage{graphicx}
\usepackage{xcolor}
\usepackage{color, colortbl}
\usepackage{marvosym}
\usepackage{multirow}
\usepackage{mdframed}
\usepackage{pgfplots,tikz,subfigure}
\usetikzlibrary{decorations.markings}
\definecolor{Gray}{gray}{0.9}
\usepackage{dcolumn,float}
\usepackage{widetext}
\usepackage{xcolor}
\usepackage{stackengine}
\usepackage{lipsum}
\usepackage{graphicx}
\usepackage[color]{changebar}
\cbcolor{blue}

\newcommand*\mnote[3][0pt]{%
  \if l#2\reversemarginpar\def\pointer{\blacktriangleright}%
    \def\stackalignment{r}\fi%
  \if r#2\normalmarginpar\def\pointer{\blacktriangleleft}%
    \def\stackalignment{l}\fi%
  \marginpar{%
    \topinset{%
      \scalebox{1.5}{\textcolor{blue}{$\pointer$}}}{%
      \belowbaseline[-1.5\baselineskip-#1]{%
        \stackengine%
          {-5pt}%
          {%
          \if l#2\hspace*{-6mm}\fi% 
          \fcolorbox{blue}{blue!20}{\parbox{11mm}%
            {\vspace{1pt}\footnotesize\raggedright#3}}}%
          {}%
          {O}%
          {l}%
          {F}%
          {F}%
          {S}%
        }%
      }{%
      0.8ex+#1}{-2.3ex}%
  }%
}

\newcommand*\mquestion[3][0pt]{%
  \if l#2\reversemarginpar\def\pointer{\blacktriangleright}%
    \def\stackalignment{r}\fi%
  \if r#2\normalmarginpar\def\pointer{\blacktriangleleft}%
    \def\stackalignment{l}\fi%
  \marginpar{%
    \topinset{%
      \scalebox{1.5}{\textcolor{red}{$\pointer$}}}{%
      \belowbaseline[-1.5\baselineskip-#1]{%
        \stackengine%
          {-5pt}%
          {%
          \if l#2\hspace*{-6mm}\fi% 
          \fcolorbox{red}{red!20}{\parbox{11mm}%
            {\vspace{1pt}\footnotesize\raggedright#3}}}%
          {}%
          {O}%
          {l}%
          {F}%
          {F}%
          {S}%
        }%
      }{%
      0.8ex+#1}{-2.3ex}%
  }%
}

\usepackage{enumerate} %Aufzaehlungen
\usepackage{newclude,pdfpages} %Seiten einbinden
\definecolor{ao}{rgb}{0.0, 0.5, 0.0}
\usepackage[ocgcolorlinks,colorlinks=true,linkcolor=blue,citecolor=ao]{hyperref} %Hervorhebung bzw. Einbiindung von Links / Querverweisen.
\usepackage{xparse,etoolbox,xargs}
\usepackage{lipsum}
\usepackage[colorinlistoftodos,prependcaption,textsize=tiny]{todonotes}
\newcommandx{\unsure}[2][1=]{\todo[linecolor=red,backgroundcolor=red!25,bordercolor=red,#1]{#2}}
\newcommandx{\thiswillnotshow}[2][1=]{\todo[disable,#1]{#2}}
\usepackage{amsmath,amssymb}
\usepackage{mathtools}
\usepackage{nicefrac}
\newcommand{\argmax}{\operatorname*{arg\,max}}
\newcommand{\vertiii}[1]{{\left\vert\kern-0.25ex\left\vert\kern-0.25ex\left\vert #1 
    \right\vert\kern-0.25ex\right\vert\kern-0.25ex\right\vert}}

\usepackage{algpseudocode}
\usepackage{algorithm}
\usepackage{algorithmicx}

\usepackage{amsthm}
\usepackage{cleveref}
\numberwithin{equation}{section}
\theoremstyle{definition}
 \newtheorem{defi}{Definition}[section]
 \newtheorem{remark}[defi]{Remark}
 
 \newtheorem{alg}[defi]{Algorithm}
 \theoremstyle{plain}
 \newtheorem{thm}[defi]{Theorem}
 
 \newtheorem{prop}[defi]{Proposition}
 \newtheorem{lemma}[defi]{Lemma}

\usepackage[noadjust]{cite}
\newcommand{\sino}{\sin(\omega_\mathrm{sup})}
\newcommand{\piu}{\|\Pi_{\widehat J_1}u\|}

\begin{document}
	
	\maketitle
	\thispagestyle{plain}
	\pagestyle{plain}

	\begin{abstract}
		In many applications, one is faced with an inverse problem, where the known signal depends in a bilinear way on two unknown input vectors. Often at least one of the input vectors is assumed to be sparse, i.e., to have only few non-zero entries. Sparse Power Factorization (SPF), proposed by Lee, Wu, and Bresler, aims to tackle this problem. They have established recovery guarantees for a somewhat restrictive class of signals under the assumption that the measurements are random. We generalize these recovery guarantees to a significantly enlarged and more realistic signal class at the expense of a moderately increased number of measurements.
	\end{abstract}
	
	%\begin{keywords}
	%Compressed Sensing, Sparse rank-one signal, Sparse Power Factorization, Required measurements 
	%\end{keywords}
	
	\section{Introduction}
	Many measurement operations in signal and image processing as well as in communication follow a bilinear model. Namely, in addition to the measurements depending linearly on the unknown signal, also certain parameters of the measurement procedure enter in a linear fashion. Hence one cannot employ a linear model (for example, in connection compressed sensing techniques \cite{candes2006robust}) unless one has an accurate estimate of these parameters.

	When such estimates are not available or too expensive to obtain, there are certain asymmetric scenarios when one of the inputs can be recovered even though the other one is out of reach (e.g., \cite{xu1995least,lee2017spectral}, this scenario is sometimes referred to as passive imaging). In most cases, however, the natural aim will be to recover both the signal and the parameters, that is, to solve the associated bilinear inverse problem. Even when some estimates of the parameters are available, such a unified approach will be preferred in many situations, especially when information is limited. Consequently, the study of bilinear inverse problems, including but not limited to the important problem of blind deconvolution, has been an active area of research for many years \cite{Haykin1994}.
	
	Observing that bilinear maps admit a representation as a linear map in the rank one outer product of the unknown signal and the parameter vector, one can approach such problems using tools from the theory of low-rank recovery (see, e.g., \cite{ahmedrechtromberg,lingstrohmer,jung2017blind}). Under sparsity assumptions, that is, when the signals and/or parameter vectors admit an approximate representation using just a small (but unknown) subset of an appropriate basis (for more details regarding when such assumptions appear in bilinear inverse problems, see  \cite{strohmer}), however, the direct applicability of these approaches is limited, as two competing objectives arise: one aims to simultaneously minimize rank and sparsity. As a consequence, the problem becomes considerably more difficult; Oymak et al., for example, have demonstrated that minimizing linear combinations of the nuclear norm (a standard convex proxy for the rank) and the $\ell_1$ norm (the corresponding quantity for sparsity) exhibits suboptimal scaling \cite{oymak2015simultaneously}.	In fact it is not even clear if without additional assumptions efficient recovery is at all possible for a near-linear number of measurements (as it would be predicted identifiability considerations \cite{doi:10.1137/16M1067469}).
	
	Recently, a number of nonconvex algorithms for bilinear inverse problems have been proposed. For example, for such problems without sparsity constraints several such algorithms have been analyzed for blind deconvolution and related problems\cite{li2016rapid,ling2017regularized} with near-optimal recovery guarantees. In contrast, our understanding of bilinear inverse problems with sparsity constraints is only in its beginning. Recently, several algorithms have been analyzed for sparse phase retrieval  \cite{bahmani2017anchored,soltanolkotabi2017structured} or blind deconvolution with sparsity constraints \cite{qu2017convolutional}. The recovery guarantees for these algorithms, however, are either suboptimal in the number of necessary measurements or only local convergence guarantees are available, i.e., one relies on the existence of a good initialization. (A noteworthy exception are the two related papers \cite{bahmani2016near,iwen2017robust}, where a two-stage approach for (sparsity) constrained bilinear inverse problems is proposed, which achieves recovery at near-optimal rate. However, the algorithm relies on a special nested structure of the measurements, which is not feasible for many practical applications.)
	
	In \cite{paper2} Lee, Wu, and Bresler  introduced the {\em sparse power factorization} (SPF) method together with a tractable initialization procedure based on  alternating minimization. They also provide a first performance analysis of their method for random bilinear measurements in the sense that their lifted representation is a matrix with independent Gaussian entries.
%	Recently, a number of nonconvex algorithms for bilinear inverse problems under sparsity constraints have been proposed. Lee, Wu, and Bresler \cite{paper2} introduced the {\em sparse power factorization} (SPF) method  based on  alternating minimization. They also provide a first performance analysis of their method for random bilinear measurements in the sense that their lifted representation is a matrix with independent Gaussian entries.
%	Recently, a number of nonconvex algorithms for bilinear inverse problems under sparsity constraints have been proposed. Lee, Wu, and Bresler \cite{paper2} introduced the {\em sparse power factorization} (SPF) method  based on  alternating minimization. They also provide a first performance analysis of their method for random bilinear measurements in the sense that their lifted representation is a matrix with independent Gaussian entries.
	That is, they work with linear operators $\mathcal{A}\colon \mathbb{C}^{n_1 \times n_2} \longrightarrow \mathbb{C}^{m} $ that admit a representation as
	\begin{equation*}
	\left(  \mathcal{A} \left( X \right)   \right) \left( \ell \right)= \text{trace} \left( A_{\ell}^\ast X  \right)
	\end{equation*}
	for i.\,i.\,d.\ Gaussian matrices $ A_{\ell}  \in \mathbb{C}^{n_1 \times n_2} $.

	For such measurements they show that with high probability, SPF converges locally to the right solution, i.e., one has convergence for initializations not too far from the signal to be recovered.

	For signals that have a very large entry, they also devise a tractable initialization procedure -- they call it thresholding initialization -- such that one has global convergence to the right solution. Local convergence has also been shown for the multi-penalty approach {\em A-T-LAS$_{1,2}$} \cite{fornasier2018}, but to our knowledge, comparable global recovery guarantees are not available to date. This is why we focus on SPF in this paper, using the results of \cite{paper2} as our starting point.
	
	The precise condition for their guarantee to hold is that both (normalized) input signals need to be larger than some $c>\tfrac{1}{2}$ in supremum norm -- more than one quarter of its mass needs to be located in just one entry, that is, the signals must have a very high peak to average power ratio.
	
	 In this paper, we considerably weaken this rather strong restriction in two ways. Firstly, we show that similar results hold for smaller lower bounds $c$ at the expense of a moderately increased number of measurements. Secondly, we show that similar results can be obtained when the mass of one of the signals is concentrated in more than one, but still a small number of entries.

	The SPF algorithm, the thresholding initialization, and the resulting recovery guarantees are reviewed in \Cref{spfsection} before we discuss and prove our results in \Cref{resultsection} and Section \Cref{sectionproof}.
	
	\subsection*{Notation}
	Throughout the paper we will use the following notation. By $ \left[n\right] $ we will denote the set $ \left\{1; \ldots; n \right\} $. For any set $J$ we will denote its cardinality by $ \vert J \vert $. For a vector $v\in \mathbb{C}^m$ we will denote by $ \Vert v \Vert $ its $ \ell_2$-norm and by $ \Vert v \Vert_{\infty}$ the modulus of its largest entry. If $J \subset \left[n\right] $ we will by $v_J$ denote the restriction of $v$ to elements indexed by $J$. For matrices $ A \in \mathbb{C}^{n_1 \times n_2} $ we will denote by $ \Vert A \Vert_F$ its Frobenius norm and by $\Vert A \Vert $ its spectral norm, i.e., the largest singular value of $A$. 	
	
	\section{Sparse Power Factorization: Algorithm and Initialization}\label{spfsection}
	
	\subsection{Problem formulation}
%	As already mentioned, every bilinear map $B: \mathbb{C}^{n_1}\times \mathbb{C}^{n_2} \rightarrow \mathbb{C}^{m}$ can be represented by a linear map
%	$\mathcal{A}\colon\mathbb{C}^{n_1\times n_2}\longrightarrow\mathbb{C}^m$ 
%	that is implicitly given by
%	\[
%	B(u,v)= \mathcal{A}(uv^\ast).
%	\]%
%	From now on we will work with this lifted representation, aiming to recover a sparse rank one solution $X$ from linear measurements $b = \mathcal{A}(X)$. Then, up to normalization, the signals $u$ and $v$ to be recovered are the 
%	left and right singular vectors of $X$.
%	Their sparsities, i.e. the number of their non-zero entries, are denoted by $s_1 $ and $s_2 $.
%	
%	
%	The main idea of sparse power factorization is to alternate between keeping $u$ fixed and solving the resulting underdetermined linear system and doing the same for $v$. 
	
	%For the linear operator $\mathcal{A}\colon \mathbb{C}^{n_1 \times n_2} \longrightarrow \mathbb{C}^{m} $ we note, that the $\ell$th entry of $ \mathcal{A} \left(X\right) $ may be represented as
	%\begin{equation}\label{operatorrepresentation}
	%\left(  \mathcal{A} \left( X \right)   \right) \left( \ell \right)= \text{trace} \left( A_{\ell}^\ast X  \right)
	%\end{equation}
	%for a matrix $ A_{\ell}  \in \mathbb{C}^{n_1 \times n_2} $.
	
	Let $ b \in \mathbb{C}^m $ be given by
	\begin{equation*}
	b:= B(u,v)+z,
	\end{equation*}
	where $ B \colon \mathbb{C}^{n_1} \times \mathbb{C}^{n_2} \rightarrow \mathbb{C}^m $ is a bilinear map and $z\in \mathbb{C}^m$ is noise. Recall that one can represent the bilinear map $B \colon \mathbb{C}^{n_1}\times \mathbb{C}^{n_2} \rightarrow \mathbb{C}^{m}$ by a linear map
	$\mathcal{A}\colon\mathbb{C}^{n_1\times n_2}\longrightarrow\mathbb{C}^m$, which satisfies
	\[
	B(u,v)= \mathcal{A}(uv^\ast).
	\]
	for all vectors $u \in \mathbb{C}^{n_1}$ and all $ v \in \mathbb{C}^{n_2}$. Note that such a linear map $ \mathcal{A}$ is characterized by a (unique) set of matrices $ \left\{ A_{\ell}  \right\}^m_{\ell =1}  \subset \mathbb{C}^{n_1 \times n_2} $ such that the $\ell$th entry of $ \mathcal{A}\left( X \right)$ is given by
	\begin{equation}\label{operatorrepresentation}
	\left(  \mathcal{A} \left( X \right)   \right) \left( \ell \right)= \text{trace} \left( A_{\ell}^\ast X  \right).
	\end{equation}
	In this notation, our goal will be to reconstruct $u$ and $v$ from linear measurements given by
	\begin{equation*}
	b_{\ell} = \text{trace} \left( A^*_{\ell} uv^* \right)
	\end{equation*}
	At the core of the Sparse Power Factorization Algorithm, as introduced in \cite{paper2}, are the linear operators $F \colon \mathbb{C}^{n_2} \longrightarrow \mathbb{C}^{m \times n_1} $ and $ G \colon \mathbb{C}^{n_1} \longrightarrow \mathbb{C}^{m \times n_2} $ defined by
	
%	\begin{defi}\label{SPF.FG}
%		Let $\mathcal{A}\colon \mathbb{C}^{n_1\times n_2} \longrightarrow \mathbb{C}^{m}$ be a linear map characterized by
%		\begin{equation}\label{operatorrepresentation}
%		\left(  \mathcal{A} \left( X \right)   \right) \left( \ell \right)= \text{trace} \left( A_{\ell}^\ast X  \right)
%		\end{equation}
%		for matrices $ A_{\ell} \in \mathbb{C}^{n_1 \times n_2} $. We then define linear operators
%		$F\colon\mathbb{C}^{n_2} \rightarrow \mathbb{C}^{m\times n_1}$
%		and 
%		$G\colon\mathbb{C}^{n_1} \rightarrow \mathbb{C}^{m\times n_2}$
%		as
		$$F(y) := 
		\begin{pmatrix}
		y^\ast A_1^\ast\\
		\vdots\\
		y^\ast A_m^\ast
		\end{pmatrix},
		\quad
		G(x) := 
		\begin{pmatrix}
		x^\ast A_1\\
		\vdots\\
		x^\ast A_m
		\end{pmatrix}.$$
		A direct consequence of this definition is that $$\mathcal{A}(xy^\ast) = [F(y)]x = \overline{[G(x)]y}$$ for all $  x \in \mathbb{C}^{n_1} $ and all $ y \in \mathbb{C}^{n_2} $.
%		for $x \in \mathbb{C}^{n_1}$ and $y \in \mathbb{C}^{n_2}$.
%		In particular, we have
%		$$\mathcal{A}(xy^\ast) = [F(y)]x = \overline{[G(x)]y}.$$
%	\end{defi}

	\subsection{Sparse Power Factorization}
	
The idea of Sparse Power Factorization is to iteratively update estimates $u_t$ and $v_t$ for $u$ and $v$ in an alternating fashion.
That is, in each iteration one keeps one of $v_t$ and $u_t$ fixed and updates the respective other one by solving an (underdetermined) linear system. Solving each of these linear systems then amounts to solving a linear inverse problem with sparsity constraints. Hence, many pursuit algorithms proposed in the context of compressed sensing can be applied such as CoSaMP \cite{needell2009cosamp}, Hard Thresholding Pursuit \cite{foucart2011hard} or Basis Pursuit. In \cite{paper2} the authors used Hard Thresholding Pursuit (HTP) for their analysis and in this paper, we will also restrict ourselves to HTP. With this, the Sparse Power Factorization Algorithm reads as follows.
	
%	Solving each of these linear systems then boils down to solving a sparsity constraint linear inverse problem, so many pursuit algorithms proposed in the context of compressed sensing can be applied. The theoretical investigations in \cite{paper2} use Hard Thresholding Pursuit \cite{foucart2011hard}, subsequently the analysis has been shown \cite{bsc} to carry over to Compressive Sampling Matching Pursuit \cite{needell2009cosamp}. In the following, we use the shortcut ``PA'' to represent either of these two (and pending theoretical investigations potentially also yet other) pursuit algorithms. With this notation, we can now formulate the Sparse Power Factorization Algorithm.
	
\begin{alg}[Algorithm 1 in \cite{paper2}]\label{SPF.alg}~\\\vspace*{-4mm}
	\begin{algorithmic}[1]
		\sffamily\small
		\Require  Operator $\mathcal{A}$, Measurement $b$, Sparsity Constraints $s_1, s_2$, Initialisation $v_0$.
		\Ensure Estimate $\widehat{X}$.
		\State $t \gets 0$
		\While{stop condition not satisfied}
		\State $t \gets t + 1$
		\State $v_{t - 1} \gets \frac{v_{t - 1}}{\big\|v_{t - 1}\big\|}$\label{vt1norm}
		\If{$s_1 < n_1$}
		\State $u_t \gets \mathrm{HTP}(\mathrm{F}(v_{t - 1}), b, s_1)$
		\Else
		\State $u_t \gets \operatorname*{{arg\,min}}\limits_x \big\|b - [\mathrm{F}(v_{t - 1})]x\big\|^2$
		\EndIf
		\State $u_t \gets \frac{u_{t}}{\big\|u_{t}\big\|}$
		\If{$s_2 < n_2$}
		\State $v_t \gets \mathrm{HTP}(\mathrm{G}(u_{t}), \bar{b}, s_2)$
		\Else
		\State $v_t \gets \operatorname*{{arg\,min}}\limits_b \big\|\bar{b}- [\mathrm{G}(u_{t})]b\big\|^2$\label{sparse-else}
		\EndIf
		\EndWhile \\
		\Return $\widehat{X} \gets u_tv_t^{\ast}$
	\end{algorithmic}
\end{alg}	
\noindent The Hard Thresholding Pursuit Algorithm is defined as follows:
\begin{alg}HTP(A, b, s)\label{HTP.alg}~\\\vspace*{-4mm}
	\begin{algorithmic}[1]
		\sffamily\small
		\Require  Measurement matrix $A \in \mathbb{C}^{m \times n}$, measurement $b \in \mathbb{C}^m$, sparsity constraint $s \in \mathbb{N}$.
		\Ensure $ \hat{x} \in \mathbb{C}^n $.
		\State $t \gets 0$
		\While{stop condition not satisfied}
		\State $t \gets t+1 $
		\State $ w= x_{t-1} + A^* \left(  b-Ax_{t-1} \right) $
		\State $J \gets  \underset{J \subset \left[n\right], \ \vert J \vert =s}{\arg \max} \ \Vert w_J \Vert  $
		\State $x_t \gets \underset{x: \text{supp} \left(x\right) \subset J }{\arg \min}  \Vert Ax-b \Vert$
		\EndWhile \\
		\Return $ \hat{x} \gets x $
	\end{algorithmic}
\end{alg}

%	\begin{alg}[Algorithm 1 in \cite{paper2}]\label{SPF.alg}~\\\vspace*{-4mm}
%		\begin{algorithmic}[1]
%			\sffamily\small
%			\Require  Operator $\mathcal{A}$, Measurement $b$, Sparsity Constraints $s_1, s_2$, Initialisation $v_0$, Pursuit Algorithm $\mathrm{PA}$.
%			\Ensure Estimate $\widehat{X}$.
%			\State $t \gets 0$
%			\While{stop condition not satisfied}
%			\State $t \gets t + 1$
%			\State $v_{t - 1} \gets \frac{v_{t - 1}}{\big\|v_{t - 1}\big\|_{\ell_2}}$\label{vt1norm}
%			\If{$s_1 < n_1$}
%			\State $u_t \gets \mathrm{PA}(\mathrm{F}(v_{t - 1}), b, s_1)$
%			\Else
%			\State $u_t \gets \operatorname*{{arg\,min}}\limits_x \big\|b - [\mathrm{F}(v_{t - 1})]x\big\|_{\ell_2}^2$
%			\EndIf
%			\State $u_t \gets \frac{u_{t}}{\big\|u_{t}\big\|_{\ell_2}}$
%			\If{$s_2 < n_2$}
%			\State $v_t \gets \mathrm{PA}(\mathrm{G}(u_{t}), \bar{b}, s_2)$
%			\Else
%			\State $v_t \gets \operatorname*{{arg\,min}}\limits_y \big\|\bar{b}- [\mathrm{G}(u_{t})]y\big\|_{\ell_2}^2$\label{sparse-else}
%			\EndIf
%			\EndWhile \\
%			\Return $\widehat{X} \gets u_tv_t^{\ast}$
%		\end{algorithmic}
%	\end{alg}
	
	\subsection{Initialization}\label{Initialization}
%	It remains to find a suitable initialization.
%	If the goal were only to minimize the rank, then alternating minimization approaches in the matrix completion literature \cite{Jain} would suggest using the leading right singular vector of  $\mathcal{A}^\ast(b)$. This vector will, however, not be sparse in general. 
%	
	As for many other non-convex algorithms (e.g., \cite{Jain,candes2015phase}), the convergence properties of Sparse Power Factorization depend crucially on the choice of the starting point. In \cite{Jain,candes2015phase} the starting point is chosen via a spectral initialization. That is, one chooses the leading left- and right-singular vectors of $ \mathcal{A}^* \left(y\right) $ as the starting point. However, in order to work this approach requires that the number of measurements is at the order of $ \max \left\{n_1, n_2\right\} $, which will in general not be optimal as it does not take into account the sparsity of the vectors $u$ and $v$. One way to incorporate the sparsity assumption would be to solve the  Sparse Principal Component Analysis (SparsePCA) problem.
	\begin{equation}\label{equ:sparsePCA}
	\begin{split}
	\max \quad &  \text{Re} \left( \tilde{u}^* \mathcal{A}^* \left(y \right) \tilde{v} \right) \\
	\text{subject to} \quad &\Vert \tilde{u} \Vert_0 \le s_1, \Vert \tilde{u} \Vert =1\\
	&\Vert \tilde{v} \Vert_0 \le s_2, \Vert \tilde{v} \Vert =1,
	\end{split}
	\end{equation}
	where $\Vert \cdot \Vert_0  $ denotes the number of non-zero entries. As it was shown in \cite[Proposition III.4]{paper2}, Algorithm \ref{SPF.alg}, if initialized by a solution of \eqref{equ:sparsePCA} is able to recover the solution $u$ and $v$ from a number of measurements at the order of $ \left( s_1 + s_2 \right) \max \left\{ \frac{s_1}{n_1}, \frac{s_2}{n_2}  \right\} $. However, the SparsePCA problem has been shown to be NP-hard \cite{tillmann2014computational}. Nevertheless, in the last fifteen years there has been a lot of research on the SparsePCA problem and, in particular, on tractable (i.e., polynomial-time) algorithms, which yield good approximations to the true solution. Several computationally tractable algorithms have been proposed for solving (\ref{equ:sparsePCA}), e.g., thresholdings algorithms \cite{ma2013sparse}, a general version of the power method  \cite{journee2010generalized} and semidefinite programs \cite{d2008optimal}. From the statistical perspective, a particular emphasis has been put for computationally efficient or at least tractable algorithms on the analysis of the single spike model\cite{amini2009high,krauthgamer2015semidefinite,deshpande2014sparse}. These approaches, however, require that the number of samples scales with the square of the number of non-zero entries of the signal to estimate (up to $\log$-factors). This raised the question whether there are fundamental barriers preventing the SparsePCA problem to be solved in polynomial time at a sampling rate close to the information theoretic limit. Indeed, it has been shown that an algorithm, that achieves this, would also allow for an algorithm which solves the $k$-clique problem in polynomial time \cite{berthet2013optimal,wang2016statistical}. However, a widely believed conjecture in theoretical computer science states, that this is not the case, which indicates that this approach will not be suited for initializing bilinear recovery problems either.\\

	\noindent In this manuscript we will analyse the following initialization algorithm, which is the one proposed in \cite{paper2}. For a set $ J_1 \subset \left[n\right] $, respectively $ J_2 \subset \left[n_2\right] $ in the following we will denote by $ \Pi_{J_1} $, respectively $ \Pi_{J_2} $ the matrix, which projects a vector onto the components which belong to $ J_1$, respectively $J_2$.
	
\begin{alg}[Algorithm 3 in \cite{paper2}]\label{SPF_alginit}~\\\vspace*{-4mm}
	\begin{algorithmic}[1]
		\sffamily\small
		\Require  Operator $\mathcal{A}$, Measurement $b$, Sparsity Constraints $s_1, s_2$, 
		\Ensure Initial guess $v_0$ for $ v \in \mathbb{C}^{n_2} $.
		%\begin{enumerate}
		%\item Start with
		%$M := \mathcal{A}^\ast(b)$
		%\item  Replace each row of $\mathcal{A}^* \left(b\right) \in \mathbb{C}^{n_1 \times n_2}$ with its best $s_2$-sparse approximation.
		%\item Choose $\widehat J_1$ to contain the indices of the $s_1$ rows of $ \mathcal{A}^* \left(b\right) $ largest in $\ell_2$-norm.
		\State For all $ i \in \left[ n_1 \right] $ let $ \xi_i  $ be the $ \ell_2$-norm of the best $s_2$-sparse approximation of the $i$th row of the matrix $ \mathcal{A}^* \left(b\right)  \in \mathbb{C}^{n_1 \times n_2} $.
		\State Let $ \widehat{J_1} \subset \left[n_1\right] $ be the set of the $ s_2 $ largest elements in $ \left\{\xi_1; \xi_2; \ldots ; \xi_{n_1}  \right\} $
		\State Choose $\widehat J_2$ to contain the indices of the $s_2$ columns of $\Pi_{\widehat J_1} \mathcal{A}^* \left( b \right)$ largest in $\ell_2$ norm, i.e.,
		\begin{equation}\label{equ:defwidetilde}
		\widehat J_2  := \underset{ J \subset \left[ n_2 \right], \ \vert J \vert = s_2 }{\argmax}    \big\|\Pi_{\widehat J_1}[\mathcal{A}^\ast(b)]\Pi_{  J }\big\|_\mathrm{F}.
		\end{equation}\\
		\Return  $v_0$, the leading right singular vector of $\Pi_{\widehat{J_1}}[\mathcal{A}^{\ast}(b)]\Pi_{\widehat{J_2}}$.
		%\end{enumerate}
		%\Return $\widehat{X} \gets u_tv_t^{\ast}$
	\end{algorithmic}
\end{alg}	
	
%	\begin{enumerate}
%		%\item Start with
%		%$M := \mathcal{A}^\ast(b)$
%		%\item  Replace each row of $\mathcal{A}^* \left(b\right) \in \mathbb{C}^{n_1 \times n_2}$ with its best $s_2$-sparse approximation.
%		%\item Choose $\widehat J_1$ to contain the indices of the $s_1$ rows of $ \mathcal{A}^* \left(b\right) $ largest in $\ell_2$-norm.
%		\item For all $ i \in \left[ n_1 \right] $ let $ \xi_i  $ be the $ \ell_2$-norm of the best $s_2$-sparse approximation of the $i$th row of the matrix $ \mathcal{A}^* \left(b\right)  \in \mathbb{C}^{n_1 \times n_2} $.
%		\item Let $ \widehat{J_1} \subset \left[n_1\right] $ be the set of the $ s_2 $ largest elements in $ \left\{\xi_1; \xi_2; \ldots ; \xi_{n_1}  \right\} $
%		\item Choose $\widehat J_2$ to contain the indices of the $s_2$ columns of $\Pi_{\widehat J_1} \mathcal{A}^* \left(b\right)$ largest in $\ell_2$ norm, i.e.,
%		\begin{equation}\label{equ:defwidetilde}
%		\widehat J_2  := \underset{ J \subset \left[ n_2 \right], \ \vert J \vert = s_2 }{\argmax}    \big\|\Pi_{\widehat J_1}[\mathcal{A}^\ast(b)]\Pi_{  J }\big\|_\mathrm{F}.
%		\end{equation}
%		\item  Output the leading right singular vector $v_0$ of $\Pi_{\widehat{J_1}}[\mathcal{A}^{\ast}(b)]\Pi_{\widehat{J_2}}$.
%	\end{enumerate}

	\section{Previous results}
	
	In the following we will work with the that the model (\ref{operatorrepresentation}), i.e., we observe
	\begin{equation*}
	\text{trace} \left( A_{\ell}^* uv^* \right) + z_{\ell}
	\end{equation*}
	where $u \in \mathbb{C}^{n_1}$ is $s_1$-sparse, $ v \in \mathbb{C}^{n_2}$ is $s_2$-sparse, and $z \in \mathbb{C}^m $ is noise. As in \cite{paper2}, $ \nu \left(z\right) $ will quantify the Noise-to-Signal Ratio by 
	\begin{equation}\label{def:noiselevel}
	\nu \left(z\right) := \frac{\Vert z \Vert}{\Vert \mathcal{A} \left(uv^*\right) \Vert }.
	\end{equation}
	\noindent For our analysis, $ \mathcal{A} $ will be a Gaussian linear operator, that is, all the entries of the matrices $ A_1, \ldots, A_{m}$ are independent with distribution $ \mathcal{CN} \left(0, \frac{1}{m} \right) $. (Here a complex-valued random variable $X$ has distribution $ \mathcal{CN} \left(0,\frac{1}{m}\right) $ if its real and complex part are independent Gaussians with expectation $0$ and variance $\sqrt{\frac{\sigma}{2}} $.)

	%The following theorem shows that SPF via thresholding initialization allows the recovery of $X$ from a number of Gaussian measurements that scales linearly in the sparsity provided $u$ and $v$ are somewhat ``peaky'', as quantified by the ratio $ \frac{\Vert u \Vert_{\infty}}{\Vert u \Vert_{\ell_2}} $.
	
	\noindent In \cite{paper2}, the authors derived that Algorithm \ref{SPF.alg}, if initialized by Algorithm \ref{SPF_alginit}, is able to recover both $u$ and $v$ (up to scale ambiguity), if both $u$ and $v$ belong to a certain restricted class of signals. More precisely, they proved the following result.
	%that $\Vert u \Vert_{\infty}\ge c \Vert u \Vert $ and $ \Vert v \Vert_{\infty} \ge c \Vert v \Vert $ for some constant $c > \frac{1}{2}$. However, while for certain applications it is realistic to assume that some fraction of the mass is concentrated on few entries, the assumption given in \cite{paper2} is far too restrictive for most applications. In \cite{geppert2017refined} we could significantly relax this assumption (in the noiseless scenario) at the cost of a slightly increased number of measurements:
	\begin{thm}[{\cite[see Theorems III.7 and Theorem III.10]{paper2}}]\label{th1}
		Assume that $\mathcal{A} \colon \mathbb{C}^{n_1 \times n_2 } \longrightarrow \mathbb{C}^{m} $ is a Gaussian linear operator as described above. Let $ b= \mathcal{A} \left( uv^* \right) + z$, where $u$ is $s_1$-sparse and $v$ is $s_2$-sparse. Suppose that $ \Vert u \Vert_{\infty} \ge 0.78 \Vert u \Vert  $, $  \Vert v \Vert_{\infty} \ge 0.78 \Vert v \Vert $, and that the noise level satisfies $ \nu \left(z\right) \le 0.04 $. Then, with probability exceeding $ 1- \exp \left(-c_1 m\right) $, the output of the Algorithm \ref{SPF.alg}, initialized by Algorithm \ref{SPF_alginit}, converges linearly to $uv^*$ provided that
	\begin{equation*}
		m \ge c_2 \left(s_1 + s_2 \right) \log \left(\max\left\{\frac{n_1}{s_1},\frac{n_2}{s_2}\right\}\right),
		\end{equation*}	
	where $c_1,c_2>0$ are absolute constants.
	\end{thm}	
	\noindent Note that in order to apply Theorem \ref{th1} to signals $u$ and $v$ one needs to require that more than half of the mass of $u$ and $v$ are located in one single entry, which is a severe restriction, which can be prohibitive for many applications. Our goal in the following will be to considerably relax this assumption by slightly increasing the amount of required measurements. We will relax this assumption in two different ways: On the one hand we will show that one can replace $ 0.78$ by an arbitrary small constant that will then show up in the number of measurements. On the other hand we generalize the result to the case that a significant portion of mass of $u$ is concentrated on a small number of entries $k$, rather than just one of them.
	\section{Main Result}\label{resultsection}
	In this section we will state the main result of this article, Theorem \ref{thm:mainresultreadable}. % For that, we will need the following definition. For any $ x \in \mathbb{C}^{n_1} $ we define
	For that, we need to define the norm
	\begin{equation*}
	\Vert x \Vert_{\left[k\right]} := \underset{I \subset \left[n\right], \ \vert I \vert = k}{\max}  \left(  \sum_{i \in I} \vert x_i \vert^2   \right)^{1/2} = \left(  \sum_{i=1}^{k}  \left( x^*_i \right)^2 \right)^{1/2},
	\end{equation*}
	for any $x\in \mathbb{C}^{n_1} $, where $ \left( x^*_i \right)^{n_1}_{i=1} $ denotes the non-increasing rearrangement of $ \left( \vert x_i \vert \right)^{n_1}_{i=1} $. Our main requirement on the vector $u$ will be that a significant amount of its mass is located in the largest $k$ entries, i.e., that $ \frac{\Vert u \Vert_{\left[k\right]}}{\Vert u \Vert} $ is large enough. %Then, our main result reads as follows:
	\begin{thm}\label{thm:mainresultreadable}
		%Let $b= \mathcal{A} \left( uv^* \right) + z$, where $u$ is $s_1$-sparse, and $v$ is $s_2$-sparse. 
		Let $ k \in \left[ n_1 \right] $ and $ 0<\xi<1, 0<\mu<1$. Then, there are absolute constants $C_1, C_2, C_3 >0$ such that if
		\begin{equation}
		m \ge C_1  \max \left\{   \frac{1}{\xi^4 \mu^4}, \frac{k}{\xi^2}  \right\}  \left(s_1 + s_2 \right) \log \left(\max\left\{\frac{n_1}{s_1},\frac{n_2}{s_2}\right\}\right),
		\end{equation}
		then with probability at least $ 1-\exp \left( - C_2   m \right)  $ the following holds.\\
		%For all $s_1$-sparse $u\in \mathbb{C}^{n_1}$, $s_2$-sparse $v \in \mathbb{C}^{n_2}$, $z \in \mathbb{C}^m$, and $ b = \mathcal{A} \left(uv^* \right) + z $ such that 
		%\begin{enumerate}
		%	\item $ \Vert u \Vert_{\left[k\right]} \ge \xi \Vert u \Vert $,
		%	\item 	  $ \Vert v \Vert_{\infty} \ge \mu \Vert v \Vert $, 
		%	\item and $ \nu\left(z\right) \le C_3 \min \left\{   \xi^2 \mu^2  ; \frac{\xi}{ \sqrt{k}}  \right\}   $,
		%\end{enumerate}
		
		For all $s_1$-sparse $u\in \mathbb{C}^{n_1}$ with $ \Vert u \Vert_{\left[k\right]} \ge \xi \Vert u \Vert $, all $s_2$-sparse $u\in \mathbb{C}^{n_2}$ with $ \Vert v \Vert_{\infty} \ge \mu \Vert v \Vert $, and all $ z \in \mathbb{C}^m $ with  $ \nu\left(z\right) \le C_3 \min \left\{   \xi^2 \mu^2  ; \frac{\xi}{ \sqrt{k}}  \right\}   $ the iterates $\{X_t\}_{t\in\mathbb{N}}$ generated by applying Algorithm \ref{SPF.alg}, initialized by Algorithm \ref{SPF_alginit},  satisfy
		$$\limsup_{t\to\infty} \frac{\|X_t - uv^*  \|_\mathrm{F}}{\| uv^* \|_\mathrm{F}} \leq 8.3 \nu.$$
		Furthermore, the convergence is linear, i.e., for all $ t \gtrsim   \log \left( \frac{1}{\varepsilon}\right) $ we have that
		\begin{equation}\label{equ:linearconvergence}
		\frac{\| X_{t} - uv^* \|_\mathrm{F}}{\| uv^*  \|_\mathrm{F}} \leq 8.3 \nu + \varepsilon.
		\end{equation}
	\end{thm}

In the following we will discuss some important special cases of Theorem \ref{thm:mainresultreadable}.
\begin{itemize}
\item \textbf{Peaky signals: } In \cite{paper2} the authors discuss recovery guarantees for signals $u$ and $v$ with $\tfrac{\Vert u \Vert_{\infty} }{\Vert u \Vert} $ and  $\tfrac{\Vert v \Vert_{\infty} }{\Vert v \Vert}  $, both bounded below by an absolute constant $\mu  \approx 0.78$. The case $k=1$ of our theorem yields a direct improvement of this result in the sense that $\mu$ can be chosen arbitrarily small with the number of required measurements only increasing by a factor of order $ \mu^{-8} $. Hence, even when this constant decays logarithmically in the dimension, the required number of measurements will only increase by logarithmic factors.
\item \textbf{Signals with multiple large entries: } When one of the input signals has multiple large entries, using the $\Vert \cdot \Vert_{[k]} $ norm improves upon the resulting guarantee as compared to the scenario just discussed. As an example, assume that $s_1=s_2=s $, that $u$ and $v$ are normalized with $\|v\|_\infty \geq c_1 s^{-1/8}$, and that $k=c_2 s^{1/2}$ of the entries of $u$ are of absolute value at least $ c_3s ^{-1/4}$. Then $ \Vert u \Vert_{\left[k\right]} \ge \sqrt{ c_2 } c_3 $. Using Theorem \ref{thm:mainresultreadable} we obtain that the vectors $u$ and $v$ can be recovered if the number of measurements is on the order of $ s^{3/2}$, thus below the order of $s^2$ that has been established for arbitrary sparse signals in \cite{strohmer} (cf. next item). In contrast, applying Theorem \ref{thm:mainresultreadable} with $k=1$ would yield that the number of measurements would have to be on the order of $   s^{5/2} $, which is worse than the state-of-the-art.
\item \textbf{Arbitrary sparse signals:}	Applying Theorem \ref{thm:mainresultreadable} to non-peaky signals yields suboptimal results. Indeed, let $u \in \mathbb{C}^{n_1}$ $s_1$-sparse and $v \in \mathbb{C}^{n_2} $ $s_2$-sparse be generic vectors. Observe that $ \Vert v \Vert_{\infty} \asymp \frac{1}{\sqrt{s_2}} \Vert v \Vert $.
%Let $J_2 \subset \left[ n_2 \right] $ be a set of cardinality $s_2$, which contains the support of $v$. Set $  \eta = \frac{\tilde{c}}{\sqrt{s_1}}  $, where $\tilde{c}$ is a numerical constant $ \tilde{c}$ chosen small enough such that $	\big\|\Pi_{ \widetilde J_1 } u   \big\|_\mathrm{\ell_2}  \ge \frac{1}{2}  $. (For example, we could choose $ \tilde{c} = \frac{1}{2}$) By choosing $\tilde{c}$ smaller, if necessary, we observe that inequality (\ref{ineq:peakinessassumption}) is fulfilled, if the noise level $\nu$ satisfies $ \nu \le \eta  $. 
Consequently, Theorem \ref{thm:mainresultreadable} applied with $ \xi =1 $, $ k= s_1 $, and $ \mu = \frac{1}{\sqrt{s_2}} $ yields that with high probability a generic $s_1$-sparse $u$ and a generic $s_2$-sparse $v$ can be recovered from $ y = \mathcal{A} \left( uv^* \right) +z $, if the number of measurements satisfies
\begin{equation*}
m \ge C \max \left\{ s_1; s_2^2 \right\} \left(s_1 + s_2 \right) \log \left(\max\left\{\frac{n_1}{s_1},\frac{n_2}{s_2}\right\}\right),
\end{equation*}	
and if the noise level $\nu$ is on the order of $ \mathcal{O} \left( \max \left\{ \frac{1}{s_2}; \frac{1}{\sqrt{s_1}}  \right\} \right) $. Previous results (see, e.g., \cite{strohmer}), in contrast, require $ m \ge C \max \left\{ s^2_1; s^2_2 \right\}  \log \left(\max\left\{\frac{n_1}{s_1},\frac{n_2}{s_2}\right\}\right) $ samples.
%\item Fix $k$. Then Theorem \ref{thm:mainresultreadable} states that $ m \gtrsim k (s_1 +s_2) \log \left( \max \left\{ n_1/s_1; n_2/s_2 \right\} \right)   $ measurements are needed to recover the signal (with high probability).
\end{itemize}

\begin{remark}
The peakiness assumptions in Theorem \ref{thm:mainresultreadable} may seem arbitrary at first sight but in certain applications they are reasonable. Namely, when $u$ is the signal transmitted via a wireless channel and $v$ is the unknown vector of channel parameters it is natural to assume that $v$ has a large entry, as the direct path will always carry most of the energy. The signal $u$ can be modified by the sender, so some large entries can be artificially introduced. In this regard, being able to consider multiple entries of comparable size is of advantage as adding a single very large entry will result in a dramatic increase of the peak-to-average power ratio.
\end{remark}

%	\begin{cor}\label{nonoise}
%		Assume that $\mathcal{A}\colon \mathbb{C}^{n_1 \times n_2 } \longrightarrow \mathbb{C}^{m} $ is a Gaussian operator. Let $ b= \mathcal{A} \left( X \right) $, where $X= uv^{\ast} \in \mathbb{C}^{n_1 \times n_2} $ is a rank-one matrix, $u$ is $s_1$-sparse, and $v$ is $s_2$-sparse. Assume $ \Vert u \Vert_{\infty} \ge \eta \Vert u \Vert_{\ell_2}  $ and $  \Vert v \Vert_{\infty} \ge \eta \Vert v \Vert_{\ell_2} $ for some $\eta>0$. Then, the SPF algorithm, initialized as described in Section \ref{Initialization}, recovers $X$ with high probability provided that
%		\begin{equation*}
%		m \ge C \eta^{-8} \left(s_1 + s_2 \right) \log \left(\max\left\{\frac{n_1}{s_1},\frac{n_2}{s_2}\right\}\right),
%		\end{equation*}	
%		where $C>0$ is a numerical constant.
%	\end{cor}

%\Cref{thm:mainresultateasy} can be deduced from Theorem \ref{thm:mainresultat} in the following way: Set $ \widetilde J_2 = \left\{\tilde{j} \right\} $, where $ \tilde{j}= \underset{j \in \left[n_2\right]}{\argmax} \ \vert v_{j} \vert$ and note that condition \ref{ineq:peakinessassumption} reformulates to $\big\|\Pi_{ \widetilde J_1 } u   \big\|_\mathrm{\ell_2}  \big\|   v   \big\|_\mathrm{\infty}  > c \max \left\{ \eta, \nu \right\}$, which finishes the proof.

\section{Proofs}\label{sectionproof}
	\subsection{Technical tools}
	The goal of this section is to prove Theorem \ref{thm:mainresultreadable}. %In the following we will always assume that $ b= \mathcal{A} \left(uv^*\right) + z $ and that $ \nu = \nu \left(z\right) $ is defined as in (\ref{def:noiselevel}).  %Note that it is sufficient to verify the claim for all $u \in \mathbb{C}^{n_1}$ and $v\in \mathbb{C}^{n_2}$ such that $ \Vert u \Vert = \Vert v \Vert =1 $. Hence, we will in the following always assume that $u$ and $v$ satisfy $ \Vert u \Vert = \Vert v \Vert =1 $.
	We will start by recalling the following variant of the well-known restricted isometry property.
	\begin{defi}[see \cite{paper2}]
	A linear operator $ \mathcal{A}$ has the $(s_1, s_2, r )$-restricted isometry property with constant $\delta$ if 
	\begin{equation}
	\left(1-\delta\right) \Vert X \Vert_F^2 \le \Vert \mathcal{A} \left( X \right)  \Vert^2 \le \left( 1+\delta \right) \Vert  X  \Vert^2_F
	\end{equation}
	for all matrices $ X \in \mathbb{C}^{n_1 \times n_2}$ of rank at most $r$ with at most $s_1$ non-zero rows and at most $s_2$ non-zero columns.
	\end{defi}
	\noindent The following lemma tells us that this property holds with high probability for a number of measurements close to the information-theoretic limit.
	\begin{lemma}[See, e.g., Theorem III.7 in \cite{paper2}]\label{thm37}
	There are absolute constants $c_1, c_2>0 $, such that if
	\begin{equation}\label{necessarymeasurementsrip}
	m \ge \frac{c_1}{\delta^2}  r \left(s_1 + s_2  \right) \log \left(  \max \left\{ \frac{n_1}{s_1}, \frac{n_2}{s_2}  \right\}   \right),
	\end{equation}
	for some $\delta >0$, then with probability at least $1-\exp \left( -  c_2 m \right) $  $ \mathcal{A}$ has the $(s_1, s_2, r)$-restricted isometry property with restricted isometry constant $\delta$.
\end{lemma}
\noindent As in \cite[Lemma VIII.7]{paper2} we will need the following quantity, which depends on $ \delta$ and $ \nu $.
	\begin{align}\label{def:omegasup}
			&\omega_\mathrm{sup} :=\sup\left\{\omega\in[0,\tfrac{\pi}{2}):\omega\geq\arcsin\left(C_\delta[\delta\tan(\omega)+(1+\delta)\nu\sec(\omega)]\right)\right\}
	\end{align}
	Here, the constant $C_\delta $ is given by the expression
	\begin{equation*}
	C_{\delta} =  1.1 \frac{ \sqrt{ \frac{2}{1-\delta^2} }  + \frac{1}{1-\delta}}{1- \sqrt{ \frac{2}{1-\delta^2} } \delta },
	\end{equation*}
	as it can be seen by an inspection of the proof of Lemma VIII.1 in \cite{paper2}. The precise value of $ C_{\delta}$ will not be important in the following, we will only use that $ 2 \le C_{\delta} \le 5 $ for $\delta \le 0.04 $. \\
	
\noindent A simple estimate for $ \omega_{\sup}$ is given by the following lemma.
\begin{lemma}\label{sin05}
	Assume that $ \delta \le 0.04 $ and $\nu \le 0.04 $. Then it holds that 
	\begin{equation*}
		\tfrac{1}{2} \leq \sino \leq 1.
	\end{equation*}
\end{lemma}
	
\begin{proof}
	%In order to prove $ \frac{1}{2} \le \sin \left( \omega_{\sup} \right)  $ 
	We observe that in order to show the claim it is enough to verify that $ \omega= \arcsin \frac{1}{2}  $ fulfills the inequality in (\ref{def:omegasup}). Indeed, using $ \cos \omega = \sqrt{\frac{3}{4}} $ and $ C_{\delta} \le 5 $ we obtain that
	\begin{align*}
		C_{\delta} \left[    \delta \tan \left(  \arcsin \frac{1}{2}   \right) + \left(1+ \delta\right) \nu   \sec \left( \arcsin \frac{1}{2}  \right)  \right] &= C_{\delta} \left[ 0.04 \frac{1/2}{\sqrt{3/4}} + \frac{ 1.04 \cdot 0.04 }{\sqrt{3/4}}   \right]  \\		
		&\le \frac{1}{2}.
	\end{align*}
	%This inequality is fulfilled due to trigonometric identities, our assumptions on $\delta$ and $\nu$, and due to $  C_{\delta} \in \left[2,5\right] $.
\end{proof}
	
	\noindent The quantity $ \omega_{\sup} $ controls the maximal angle between the initialization $ v_0$ and the ground truth $v$ such that the Sparse Power Factorization is guaranteed to converge as captured by the following theorem.
%	\begin{defi}[see Lemma VIII.7 in \cite{paper2}]\label{osup}
%		Depending on the restricted isometry constant $\delta$, we denote
%		\begin{align}\label{def:omegasup}
%		&\omega_\mathrm{sup} :=\sup\left\{\omega\in[0,\tfrac{\pi}{2}):\omega\geq\arcsin\left(C_\delta[\delta\tan(\omega)+(1+\delta)\nu\sec(\omega)]\right)\right\}
%		\end{align}
%		where $C_\delta \in [2,5]$ (for $\delta \leq 0.04$) is a constant that depends on $\delta$ and the choice of pursuit algorithm.\textbf{TODO: Reference for $C_{\delta}$}, \textbf{Remark: arcsin(...) is an increasing function of $\omega$!}
%	\end{defi}	
		\begin{thm}[Theorem III.9 in \cite{paper2}]\label{thm39}
		Assume that
		\begin{enumerate}[1)]
			% \item $X$ satisfies the sparsity conditions $\|u\|_0\leq s_1$ and $\|v\|_0 \leq s_2$,
			\item $\mathcal{A}$ has the $(3s_1,3s_2,2)$-RIP with isometry constant $\delta\leq0.08$,
			\item $\nu \leq 0.08$,
			\item the initialization $v_0$ satisfies $\sin(\angle(v_0,v)) < \sin \left( \omega_{\sup}  \right)$.
		\end{enumerate}
		Then the iterates $\{X_t\}_{t\in\mathbb{N}} $ generated by Algorithm \ref{SPF.alg}, initialized via Algorithm \ref{SPF_alginit}, satisfy %converges  starting from the thresholding initialization and the iterates $\{X_t\}_{t\in\mathbb{N}}$ satisfy
		$$\limsup_{t\to\infty} \frac{\|X_t - uv^*\|_\mathrm{F}}{\|uv^*\|_\mathrm{F}} \leq 8.3 \nu.$$
		Furthermore, the convergence is linear in the sense of (\ref{equ:linearconvergence}).
	\end{thm}
	\noindent Thus, it remains to verify that the initialization satisfies $\sin(\angle(v_0,v)) < \sin \left( \omega_{\sup}  \right)$. The following lemma gives an upper bound on $\sin(\angle(v_0,v))$.	
		\begin{lemma}[Lemma 8 in \cite{paper2}]\label{lemma8.10}
		Assume that the $(3s_1, 3s_2, 2 )$-restricted isometry property holds for some constant $ \delta >0 $. Furthermore, assume that $ \Vert u \Vert = \Vert v \Vert =1 $. Let $\widehat{J_1} \subseteq \left[n_1\right] $ and $\widehat{J_2} \subseteq \left[n_2\right]$ denote the output resulting from Algorithm \ref{SPF_alginit}.
		 %such that $J_1$, respectively $J_2$, has cardinality at most $s_1$, respectively $s_2$.
		 Denote by $v_0$ the leading right singular vector of $\Pi_{\widehat{J_1}}[\mathcal{A}^\ast(b)]\Pi_{\widehat{J_2}}$. Then it holds that
		\begin{equation}\label{ineq:sufficientcondition2}
		\sin(\angle(v_0,v)) \leq \frac{\big\|\Pi_{\widehat J_1}u\big\|\big\|\Pi_{\widehat J_2}^\perp v\big\| + (\delta + \nu+\delta\nu)}{\big\|\Pi_{\widehat J_1}u\big\|-(\delta + \nu+\delta\nu)}.
		\end{equation}
	\end{lemma}

	\noindent  Furthermore, we will need the following two lemmas for our proof.
\begin{lemma}\label{lemma:lastlemma}[Lemma VIII.12 in \cite{paper2}]
Let $u$ and $v$ be as in Lemma \ref{lemma:supportlowerbound} and assume that the measurement operator $ \mathcal{A}$ satisfies the $ \left( 3s_1, 3s_2, 2 \right) $-restricted isometry property with constant $ \delta $. Recall that $\widehat J_1 \subset \left[n_1\right] $ is the support estimate for $v_0$ given by the initialization algorithm \ref{SPF_alginit}. Define
\begin{equation}\label{equ:definitionJ1}
	\widetilde{J_1} :=  \left\{  j \in \left[ n_1 \right]: \ \vert u_j \vert \ge 2 \left( \delta + \nu + \delta \nu \right)    \right\}.
\end{equation}	
Then we have that $ \widetilde{J_1}  \subset \widehat J_1  $. 
\end{lemma}	
	
\begin{lemma}\label{lemma:supportlowerbound}
	Assume that $ \mathcal{A}$ has the $(3s_1, 3s_2, 2)$-restricted isometry property with isometry constant $ \delta >0 $ and assume that $u$, respectively $v$, are $s_1$-sparse, respectively $s_2$-sparse, and satisfy $ \Vert u \Vert = \Vert v \Vert =1 $. Let $ \widetilde{J_1} $  be defined as in (\ref{equ:definitionJ1}).
	Then, it holds that
	\begin{equation*}
	\big\|\Pi_{ \widehat J_1 } u   \big\|  \big\|  \Pi_{ \widehat J_2 } v   \big\| \ge   \big\|\Pi_{ \widetilde{J_1} } u   \big\|  \Vert v \Vert_{\infty} - 2   \left( \delta + \nu +\delta \nu  \right).
	\end{equation*}
\end{lemma}

%	\begin{lemma}\label{cd75.2}
%		Assume $ \delta < 0.04 $ and $ \nu < 0.04 $. Then it holds that 
%		\begin{equation*}
%		0 \leq C_\delta(\delta + \nu+\delta\nu) < \frac{1}{2}.
%		\end{equation*}
%	\end{lemma}
%	\begin{proof}
%		Using the expanded expressions for $C_\delta$ from \cite{bsc}, we see that $C_\delta$ is non-negative for both HTP and \cosamp{}.
%		We then note, that the expressions for $C_\delta$ are monotonous in $B$.
%		By mere calculations we obtain that
%		$C_\delta^\mathrm{HTP}(\delta + 2\nu+2\delta\nu) < \tfrac{1}{2}$
%		for $B = 0.08$ and 
%		$C_\delta^{\cosampM}(\delta + 2\nu+2\delta\nu)  < \tfrac{1}{2}$
%		for $B = 0.05$.
%		So these bounds especially hold for $B < 0.04$.
%	\end{proof}
%	
\noindent Lemma \ref{lemma:supportlowerbound} is actually a slight generalization of what has been shown in \cite[p. 1685]{paper2}. For completeness we have included a proof in Section \ref{section:supportlowerbound}, which closely follows the proof in \cite{paper2}.\\

	\subsection{Proof of our main result}
		We will now piece together these ingredients to obtain a sufficient condition; in the remainder of the section we will then show that the condition holds in our measurement setup. First note that in order to apply Theorem \ref{thm39} we need to check that 	$\sin(\angle(v_0,v)) < \sin \left( \omega_{\sup}  \right)$ is satisfied. By Lemma \ref{lemma8.10} it is sufficient to show that the right-hand side of inequality (\ref{ineq:sufficientcondition2}) is strictly smaller than $ \sin \left( \omega_{\sup} \right) $. Combining this with the equality $ \big\|\Pi_{\widehat J_2}^{\perp}v\big\| =  \sqrt{ 1 - \big\|\Pi_{\widehat J_2}v\big\|^2 }  $ we obtain the sufficient condition
		\begin{equation*}
		\big\|\Pi_{\widehat J_1}u\big\|   \sqrt{ 1 - \big\|\Pi_{\widehat J_2}v\big\|^2 }  <  \sin \left( \omega_{\sup} \right)  \left(  \big\|\Pi_{\widehat J_1}u\big\|   - \left( \delta + \nu+\delta\nu\right)   \right) -  \left( \delta + \nu+\delta\nu  \right) 
		\end{equation*}
		Further manipulations yield that this is equivalent to
		\begin{equation}\label{ineq:sufficientcondition}
		\begin{split}
		\big\|\Pi_{\widehat J_1}u\big\|^2 < &\left( \sin \left( \omega_{\sup}  \right) \big\|\Pi_{\widehat J_1}u\big\|   - \left( 1+ \sin \left( \omega_{\sup} \right) \right) \left( \delta + \nu+ \delta \nu \right)    \right)^2\\
		+ & \big\|\Pi_{\widehat J_1}u\big\|^2 \big\|\Pi_{\widehat J_2} v\big\|^2.
		\end{split} 
		\end{equation}		
  		Hence, in the following our goal will be to verify (\ref{ineq:sufficientcondition}).
		\noindent We already noticed that the angle $ \omega_{\sup} $ measures how much the vector $v_0$ given by the initializiation has to be aligned with the ground truth $v$ in order for the Sparse Power Factorization to converge. Consequently, it is natural to expect that the smaller the constant $ \delta$ and the noise-to-signal ratio $\nu$, the less the initializiation vector has to be aligned with the ground truth, i.e., the larger $ \omega_{\sup} $ can be. This fact is captured by the following lemma.
	\begin{lemma}\label{sin2}
		Let $\delta \leq 0.04$ and $\nu \leq 0.04$. Then it holds that
		$$\sino \geq  1 -C_{\delta}^2\left(\delta + 2\delta\nu+2\nu\right)^2.$$
	\end{lemma}
	\begin{proof}
		It follows directly from \eqref{def:omegasup} that
		\begin{align*}
		 \omega_{\sup}   &= \arcsin \left(  C_{\delta}  \left[ \delta \tan \left( \omega_{\sup} \right) + \left( 1+ \delta \right) \nu \sec \left(  \omega_{\sup} \right)    \right] \right).
		\end{align*}
		Using trigonometric identities we obtain that
		\begin{align*}
		 \sin \left( \omega_{\sup} \right)  &= C_{\delta} \left[ \delta \frac{\sin \left(  \omega_{\sup} \right)}{\sqrt{1-\sin \left(  \omega_{\sup} \right)^2 }}  + \left(1+\delta \right) \nu \frac{1}{\sqrt{ 1- \sin \left( \omega_{\sup} \right)^2  }}    \right].
		\end{align*}
		Lemma \ref{sin05} implies that
		\begin{equation*}
		\sin \left( \omega_{\sup}  \right) \le \frac{ \sin \left( \omega_{\sup} \right) }{\sqrt{1-\sin \left(  \omega_{\sup} \right)^2}} C_{\delta} \left( \delta + 2 \left( 1+ \delta \right) \nu  \right).
		\end{equation*}
	Rearranging terms yields that
		\begin{equation*}
		\sin \left( \omega_{\sup}  \right) \ge \sqrt{1 - C^2_{\delta}  \left( \delta +2\delta \nu +2\nu   \right)^2 }.
		\end{equation*}
		The claim follows then using the fact that $ \sqrt{x} \ge x $ for all $ x \in \left[ 0,1\right] $. %This is ensured for $0<C^2_{\delta}  \left( \delta +2\delta \nu +2\nu   \right)^2 \le \frac{1}{2}$, which follows from $ C_{\delta} \in \left[2,5 \right] $ and our assumptions on $\delta $ and $\nu$.
		\end{proof}	
	\noindent With these preliminary lemmas, we can now prove the following proposition, which is a slightly more general form of Theorem \ref{thm:mainresultreadable}.
		\begin{prop}\label{prop:mainproposition}
	There are absolute constants $c_1, c_2, c_3 >0$ such that if
	\begin{equation}\label{equ:numbermeasurements}
	m \ge c_1 \delta^{-2} \left(s_1 + s_2 \right) \log \left(\max\left\{\frac{n_1}{s_1},\frac{n_2}{s_2}\right\}\right),
	\end{equation}		
	for some $  0 < \delta <0.01$, then with probability at least $ 1-\exp \left( - c_2 m \right) $ the following statement holds uniformly for all $s_1$-sparse $u \in \mathbb{C}^{n_1} $, $s_2$-sparse $v \in \mathbb{C}^{n_2} $ and $ z \in \mathbb{C}^m $ such that $\Vert u \Vert = \Vert v \Vert =1 $ and $ \nu \left(z\right) \le 0.01 $:\\	
	\noindent Let the measurements be given by $ b= \mathcal{A} \left( uv^* \right) +z $ for $ \mathcal{A} $ Gaussian as above and let $ \widetilde J_1 $ be defined by %the set defined by (\ref{equ:definitionJ1}).
	\begin{equation}\label{equ:definitionJ2}
	\widetilde J_1 :=  \left\{  j \in \left[ n_1 \right]: \ \vert u_j \vert \ge  M_{\delta,\nu} \right\},
	\end{equation}
	where
	\begin{equation*}
	M_{\delta, \nu} :=  2 \left( \delta + \nu + \delta \nu \right).
	\end{equation*}
	%\begin{equation}\label{equ:definitionJ1}
	%\widetilde J_1 :=  \left\{  j \in \left[ n_1 \right]: \vert u_j \vert \ge 2 \left( \eta + \nu + \eta \nu \right)    \right\}.
	%\end{equation} 
	%$\widetilde{J_1} $ defined in (\ref{equ:definitionJ1}) 
	%Then, if there is a set $ \widetilde J_2 \subset \left[ n_2 \right] $ such that
	Then, whenever
	\begin{equation}\label{ineq:peakinessassumption}
	%\big\|\Pi_{ \widetilde J_1 } u   \big\|_\mathrm{\ell_2}  \big\|  \Pi_{ \widetilde J_2 } v   \big\|_\mathrm{\ell_2}  > 2 \left(  \sqrt{ \min   \left\{ \vert \widehat J_1 \vert  ;\vert \widetilde J_2 \vert  \right\}   } + \sqrt{  C_{\delta} + 2}  \right) \left( \eta +2 \nu + 2\eta \nu  \right) 
	\big\|\Pi_{ \widetilde J_1 } u   \big\|  \Vert v \Vert_{\infty} > c_3 \sqrt{    M_{\delta,\nu} } ,
	\end{equation}
	the iterates $\{X_t\}_{t\in\mathbb{N}}$ generated by Algorithm \ref{SPF.alg} initialized via Algorithm \ref{SPF_alginit}, satisfy
	$$\limsup_{t\to\infty} \|X_t - uv^*\|_\mathrm{F} \leq 8.3 \nu.$$ 
	Furthermore, the convergence is linear in the sense of (\ref{equ:linearconvergence}).
\end{prop}

\begin{proof}[Proof of Proposition \ref{prop:mainproposition}]	
	%Note that we may assume without loss of generality that $ \Vert u \Vert_{\ell_2}= \Vert v \Vert_{\ell_2} =1 $. Hence, the above assumptions on $\|u\|_\infty$ and $\|v\|_\infty$ reformulate to  
	%$\|u\|_\infty \geq \eta,\quad\|v\|_\infty \geq \eta$. 	
	Assumption (\ref{equ:numbermeasurements}) and \Cref{thm37} yield that with probability at least $1 - \exp \left(-c  m \right) $ the ($3s_1$,$3s_2$,$2$)-restricted isometry property holds with constant $\delta $. %By increasing the constant $c_1$ in (\ref{equ:numbermeasurements}), if necessary, we can assume that $ \delta \le 0.01 $.
	For the remainder of the proof, we will consider the event that the restricted isometry property holds for such $\delta$. 	
	%	Using \Cref{lemma:supportlowerbound} we obtain that
	%	\begin{align*}
	%	\big\|\Pi_{ \widehat J_1 } u   \big\|^2_\mathrm{\ell_2}  \big\|  \Pi_{ \widehat J_2 } v   \big\|^2_\mathrm{\ell_2} &\ge  \left(  \big\|\Pi_{ \widetilde J_1 } u   \big\|_\mathrm{\ell_2}  \big\|  \Pi_{ \widetilde J_2 } v   \big\|_\mathrm{\ell_2}  - 2  \sqrt{ \min   \left\{ \vert \widehat J_1 \vert  ;\vert \widetilde J_2 \vert  \right\}   } \left( \delta + \nu +\delta \nu  \right) \right)^2 \\
	%	&>    \sqrt{   \min \left\{ \vert \widehat J_1 \vert  ;\vert \widetilde J_2 \vert  \right\} }  \left( c \max \left\{ \eta, \nu \right\}   -  2   \left( \delta + \nu +\delta \nu  \right) \right)^2\\
	%	& \ge 4 \left( C_{\delta} +2 \right) \left(  \delta + 2\nu + 2 \delta \nu \right)^2,
	%	\end{align*}	
	%	where the second inequality follows from assumption (\ref{ineq:peakinessassumption}) and the last one holds if the constant $c$ is chosen large enough (and also due to $ C_{\delta} \in \left[2,5 \right] $).	
	%Assumption (\ref{ineq:peakinessassumption}) yields that
	We obtain
	\begin{equation*}
	\big\|\Pi_{ \widetilde J_1 } u   \big\|  \Vert v \Vert_{\infty}  \ge  \left(    \sqrt{  C^2_{\delta} +1 }  + 1       \right)  \sqrt{M_{\delta,\nu}}
	\end{equation*}
	from $ 2\le C_{\delta} \le 5 $ and by choosing the constant $c_3$ in assumption (\ref{ineq:peakinessassumption}) large enough. Combining this with Lemma \ref{lemma:supportlowerbound}  we obtain that
	\begin{equation}\label{ineq:chain4}
	\begin{split}
	\big\|\Pi_{ \widehat J_1 } u   \big\|  \big\|  \Pi_{ \widehat J_2 } v   \big\| &\ge   \big\|\Pi_{ \widetilde{ J_1} } u   \big\|    \Vert v \Vert_{\infty}  - M_{\delta, \nu}.\\
	&>  \sqrt{  \left( C^2_{\delta} + 1  \right)        M_{\delta, \nu} },
	\end{split}
	\end{equation}
	where we used that $\sqrt{x} \ge x $ for all $ x \in \left[0,1\right] $. This yields a lower bound for the second summand of the right-hand side of (\ref{ineq:sufficientcondition}). To bound the first summand we estimate
	\begin{equation}\label{ineq:chain6}
	\begin{split}
	&\sino \piu - \left( \sino +1 \right)  \left( \delta+\nu + \delta \nu \right)\\
	\ge&   \left( 1- C^2_{\delta} \left( \delta + 2\nu + 2\delta \nu  \right)^2 \right) \piu  -2 \left( \delta + \nu +\delta \nu \right)   \\
	\ge&   \piu - C^2_{\delta} \left( \delta +2 \nu + 2\delta \nu  \right)^2 -2 \left( \delta + \nu + \delta \nu \right)   \\
	\ge &   \piu - \left( C^2_{\delta}+1 \right)   M_{\delta,\nu}  \\
	\ge & 0.
	\end{split}
	\end{equation}
	In the first line we used \Cref{sin2} and the fact that $ \sino \le 1 $. The second line is due to $ \piu \le 1  $ and the third inequality is due to $ \delta \ge 0 $, $ \nu \ge 0 $. In order to verify the last inequality it is enough to observe that due to Lemma \ref{lemma:lastlemma} and due to assumption (\ref{ineq:peakinessassumption}) with $c_3$ large enough
	\begin{align*}
	\piu  &\ge \Vert \Pi_{\widetilde{J_1}}  u \Vert \ge \Vert \Pi_{\widetilde{J_1}}  u \Vert   \big\|  \Vert v \Vert_{\infty}  \ge \left( C^2_{\delta} + 1  \right) M_{\delta,\nu},
	\end{align*}
	where the last inequality uses that $ C_{\delta} \le 5 $ and $ 0 \le \delta, \nu \le 0.01 $. Hence, by squaring (\ref{ineq:chain6}) we obtain that
	\begin{equation}\label{ineq:chain5}
	\begin{split}
	&\left( \sino \piu - \left( \sino +1 \right) \left( \delta + \nu + \delta \nu \right)  \right)^2\\
	\ge & \left(    \piu -  \frac{1}{2}\left( C^2_{\delta}+1 \right)  M_{\delta, \nu}    \right)^2\\
	\ge& \piu^2 -  \left( C^2_{\delta}+1 \right) M_{\delta, \nu}  \piu  \\
	\ge & \piu^2 -  \left( C^2_{\delta}+1 \right)  M_{\delta, \nu},
	\end{split}
	\end{equation}
	where in the last line we again used that $ \piu \le 1 $. Together with (\ref{ineq:chain4}) this yields (\ref{ineq:sufficientcondition}), as desired.
	
%	Combining the inequality chains (\ref{ineq:chain4}) and (\ref{ineq:chain5}) we obtain that
%	\begin{align*}
%	\left( \sino\piu - (\sino+1) \left( \delta + \nu + \delta \nu \right)   \right)^2 + \big\|\Pi_{ \widehat J_1 } u   \big\|^2  \big\|  \Pi_{ \widehat J_2 } v   \big\|^2
%	% \ge &\piu^2 - 2 \left( C_{\delta}+2 \right)   \left( \delta + 2 \nu +2 \delta \nu \right)^2 + 2 \delta \\
%	> & \piu^2.
%	\end{align*}
%	%where the last inequality follow from our assumptions on $ \delta $, $\nu$, and $ C_{\delta} \in \left[2,5\right] $. 
%	Thus, by Theorem \ref{thm39}, Lemma \ref{lemma8.10}, and Remark \ref{remark:connectionthmlem} the result follows.
\end{proof}
\noindent Finally, we will deduce Theorem \ref{thm:mainresultreadable} from Proposition \ref{prop:mainproposition}.
\begin{proof}[Proof of Theorem \ref{thm:mainresultreadable}]
	We will prove this result by applying Proposition \ref{prop:mainproposition} with 
	\begin{equation}
	\delta = \min  \left\{  \frac{\xi}{6 \sqrt{2k}}     ;      \frac{\xi^2 \mu^2}{8c_3^2}     \right\}.
	\end{equation}
     Let $ u \in \mathbb{C}^{n_1}$ $s_1$-sparse, $ v \in \mathbb{C}^{n_2} $ $s_2$-sparse and $z \in \mathbb{C}^m $ such that the assumptions of Theorem \ref{thm:mainresultreadable} are satisfied. Without loss of generality we may assume in the following that $\Vert  u \Vert = \Vert v \Vert =1  $.
	First, we note that invoking $ \delta, \nu < 0.01 $ and potentially decreasing the size of $C_3$ we have that
	\begin{align*}
	2 \left( \delta + \nu \left(z\right) + \delta \nu \left(z\right)   \right) < 2 \left( \delta + 2 \nu \left(z\right) \right) \le \frac{\xi}{\sqrt{2k}}.
	\end{align*}
	Hence, we obtain that
	\begin{equation}\label{equ:Jinclusion}
	\breve{J}_1:= \left\{  j \in \left[ n_1 \right]:  \ \vert u_j \vert \ge \frac{\xi}{\sqrt{2k}}   \right\}  \subset \widetilde{J}_1,
	\end{equation}
	where $ \widetilde{J}_1 $ is the set defined in (\ref{equ:definitionJ2}). 
	%We will verify first that $ \Vert \Pi_{\breve{J}_1} u \Vert \ge \frac{\xi}{ \sqrt{2} } $. If this would not be the case, we would have that $  \sum_{ i \in \left[ k \right] \backslash \breve{J}_1  }  \left(  u^*_i  \right)^2 \ge   \frac{\xi^2}{2}  $. 
	Note that
	\begin{equation*}
	\sum_{i \in \left[k\right] \backslash \breve{J}_1 } \left( u^*_i \right)^2 < \sum_{i \in \left[k\right] \backslash \breve{J}_1 } \frac{\xi^2}{2k} \le  \frac{\xi^2}{2},
	\end{equation*}
	where in the first inequality we have used that $ u^*_i < \frac{\xi}{\sqrt{2k}} $ for all $ i \in \left[k\right] \backslash \breve{J}_1 $. By the assumption $ \Vert u \Vert_{\left[k\right]} \ge \xi $ this yields that $ \sum_{i \in \left[k\right] \cap \breve{J}_1  }  \left( u^*_i \right)^2  \ge \frac{\xi^2}{2} $, which in turn implies that $  \Vert \Pi_{\breve{J}_1}  u  \Vert \ge \frac{\xi}{\sqrt{2}} $. By the inclusion (\ref{equ:Jinclusion}) we obtain that $   \Vert \Pi_{ \widetilde J_1 }  u  \Vert \ge \frac{\xi}{\sqrt{2}} $. Hence,  using the assumption $ \Vert v \Vert_{\infty} \ge \mu $, our choice of $\delta$, the assumption on the noise level $ \nu \left( z \right) $ and potentially again decreasing the value of the constant $C_3$ we obtain that
	\begin{equation*}
	\Vert \Pi_{\widetilde{J}_1}  u   \Vert \Vert v \Vert_{\infty}  \ge \frac{\xi \mu }{ \sqrt{2} } \ge c_3 \sqrt{ M_{\delta, \nu} }.
	\end{equation*}
	This shows that (\ref{ineq:peakinessassumption}) is satisfied. Hence, we can apply Proposition \ref{prop:mainproposition} and by inserting our choice of $\delta$ into (\ref{equ:numbermeasurements}), so choosing the constant $C_1$ large enough, we obtain the main result.
	
\end{proof}

	\section{Outlook}
	We see many interesting directions for follow-up work. Most importantly, it remains to explore whether additional constraints on the signals to be recovered are truly necessary (cf. our discussion on to SparsePCA in Section \ref{Initialization}). Even if this is the case, there is substantial room for improvement with respect to the noise-dependence of the recovery results. A direction to proceed could be to consider stochastic noise models instead of deterministic noise. Also in this work we exclusively considered operators $\mathcal{A}$ constructed using Gaussian matrices. However, in many applications of interest, the measurement matrices possess a significantly reduced amount of randomness. For example, in blind deconvolution one typically encounters rank-one measurements. That is, the restricted isometry property as used in this paper does not hold. Thus, one needs additional insight to study whether  there exists a computationally tractable initialization procedure at a near-optimal sampling rate. First steps in this direction were taken in \cite{lee2015rip,lee2017blind}, but a lot of questions remain open.

%	Even though our result looks quite similar to Theorem~\ref{th1}, we believe that it is an important step towards practiability. Namely, mild peakiness requirements can actually be considered realistic in certain applications. In wireless communication, for example, the channel is typically somewhat peaky as a significant portion of the energy will be transmitted via the direct path, and for the signal, peakiness can be imposed by adding a known pilot tone to the signal. Peakiness assumptions as strong as in Theorem~\ref{th1}, however, will neither be desirable for the signal nor realistic for the channel.
%	
%	The need for peakiness is also confirmed by numerical experiments, cf.\ \cite{msc}, which show that the relative frequency of successful runs of SPF decreases monotonically alongside with decreasing peakiness. 

	\section*{Acknowledgements}
%	Jakob Geppert is supported by the  German Science Foundation (DFG) in the Collaborative Research Centre ``SFB 755: Nanoscale Photonic Imaging'' and partially in the framework of the Research Training Group ``GRK 2088: Discovering Structure in Complex Data: Statistics meets Optimization and Inverse Problems'' which he both gratefully acknowledges.
%	He further thanks Gerlind Plonka for her support and fruitful discussions.	
	Jakob Geppert is supported by the  German Science Foundation (DFG) in the Collaborative Research Centre ``SFB 755: Nanoscale Photonic Imaging'' and partially in the framework of the Research Training Group ``GRK 2088: Discovering Structure in Complex Data: Statistics meets Optimization and Inverse Problems''. Felix Krahmer and Dominik St\"oger have been supported by the German Science Foundation (DFG) in the context of the joint project ``SPP 1798: Bilinear Compressed Sensing'' (KR 4512/2-1). Furthermore, the authors want to thank Yoram Bresler and Kiryung Lee for helpful discussions.

	%\appendices

	\bibliography{m-thesis}~

\appendix

		\section{Proof of Lemma \ref{lemma:supportlowerbound}}\label{section:supportlowerbound}

For the proof of Lemma \ref{lemma:supportlowerbound} we will use the following result.
\begin{lemma}\label{lemma:closeident}[Lemma A.2 and Lemma A.3 in \cite{paper2}]
	Assume that the $(3s_1, 3s_2, 2 )$-restricted isometry property is fulfilled for some restricted isometry constant $ \delta >0 $. Assume that the cardinality of $\widetilde J_1 \subseteq \left[ n_1 \right]$, respectively $\widetilde J_2 \subseteq \left[ n_2\right]$ is at most $2s_1$, respectively $ 2s_2 $.
	%Suppose that $X \in \mathbb{C}^{n_1\times n_2}$ is doubly $(s_1, s_2)$-sparse and has rank one.
	%Then for all $ u \in \mathbb{C}^{n_1} $ and $ v \in \mathbb{C}^{n_2}$ such that 
	Then, whenever $u\in \mathbb{C}^{n_1}$ is at most $2s_1$-sparse and $v \in \mathbb{C}^{n_2}$ is at most $2s_2$-sparse, we have that
	\begin{align*}
	\|\Pi_{\widetilde J_1}[(\mathcal{A}^\ast\mathcal{A}- I)(uv^*)]\Pi_{\widetilde J_2}\| \leq \delta \|uv^*\|_\mathrm{F}.
	\end{align*}
	Furthermore for all $ z \in \mathbb{C}^n $ and for all $\widetilde J_1 \subseteq \left[ n_1 \right]$, respectively $\widetilde J_2 \subseteq \left[ n_2\right]$, with cardinality at most $s_1$, respectively $ s_2 $, we have that
	\begin{align*}
	\|\Pi_{\widetilde J_1}[\mathcal{A}^\ast(z)]\Pi_{\widetilde J_2}\| \leq \sqrt{1+\delta}\|z\|_{\ell_2}.
	\end{align*}
\end{lemma}

\begin{proof}[Proof of Lemma \ref{lemma:supportlowerbound}]
 Recall that $b = \mathcal{A}\left(X\right) + z $ and define $k_1$ and $k_2$ by
\begin{equation}\label{equ:defwidehat}
\begin{split}
k_1 &:= \underset{ k \in \left[ n_2 \right] }{\argmax} \  \vert v_k \vert\\
k_2 & := \underset{ k \in \left[ n_2 \right] }{\argmax}    \big\|\Pi_{\widehat J_1}[\mathcal{A}^\ast(b)]\Pi_{  \left\{ k \right\} }\big\|_\mathrm{F}.
\end{split}
\end{equation}
%\widetilde J_2 &:=  \underset{ J \subset \left[ n_2 \right], \ \vert J \vert = k  }{\argmax} \  \big\|\Pi_{J} v \big\|_\mathrm{\ell_2}
The starting point of our proof is the observation that
\begin{equation}\label{ineq:chain1}
\big\|\Pi_{\widehat J_1}[\mathcal{A}^\ast(b)]\Pi_{ \left\{ k_2 \right\} }\big\|_\mathrm{F} \ge \big\|\Pi_{\widehat J_1}[\mathcal{A}^\ast(b)]\Pi_{  \left\{ k_1 \right\} }\big\|_\mathrm{F} \ge \big\|\Pi_{\widetilde{J_1}}[\mathcal{A}^\ast(b)]\Pi_{ \left\{ k_1 \right\} }\big\|_\mathrm{F},
\end{equation}
where the first inequality is due to the definition of $k_2$ and the second one follows from $  \widetilde{J_1}  \subset \widehat J_1 $, which is due to Lemma \ref{lemma:lastlemma}. The right-hand side of the inequality chain can be estimated from below by 
\begin{equation}\label{ineq:chain2}
\begin{split}
&\big\|\Pi_{\widetilde{J_1}}[\mathcal{A}^\ast(b)]\Pi_{ \left\{ k_1 \right\} }\big\|_\mathrm{F}\\  
\geq & \big\|\Pi_{\widetilde{J_1} } uv^* \Pi_{ \left\{ k_1 \right\} }\big\|_\mathrm{F} -  \big\|\Pi_{ \widetilde{J_1}} \left[  \left( \mathcal{A}^* \mathcal{A} - I  \right) \left(uv^*\right) \right]  \Pi_{ \left\{ k_1 \right\} }\big\|_\mathrm{F} -  \big\|\Pi_{\widetilde{J_1} } \mathcal{A}^* \left(z\right) \Pi_{ \left\{ k_1 \right\} }\big\|_\mathrm{F} \\
%\geq & \big\|\Pi_{ \overline{J_1} } uv^* \Pi_{ \widetilde J_2 }\big\|_\mathrm{F} -  \sqrt{ \min \left\{ \vert \overline{J_1} \vert ; \vert \widetilde J_2 \vert   \right\} }   \left(  \big\|\Pi_{ \overline{J_1} }  \left[  \left( \mathcal{A}^* \mathcal{A} - I  \right) \left(uv^*\right) \right]   \Pi_{ \widetilde J_2 }\big\| +  \big\|\Pi_{ \overline{J_1} } \mathcal{A}^* \left(z\right) \Pi_{ \widetilde J_2 }\big\| \right) \\
\geq & \big\|\Pi_{ \widetilde{J_1} } uv^* \Pi_{ \left\{ k_1 \right\} }\big\|_\mathrm{F} -  \left(  \delta \Vert uv^* \Vert_F  +  \sqrt{1 + \delta}  \Vert z \Vert \right) \\
\ge &    \big\|\Pi_{ \widetilde{J_1} } u   \big\|  \Vert v \Vert_{\infty} -  \left( \delta + \nu +\delta \nu  \right).
\end{split}
\end{equation}
In the first inequality we used $b= \mathcal{A} \left(uv^*\right) + z$ and the triangle inequality. The second inequality follows from Lemma \ref{lemma:closeident}. The last line follows from $ \Vert uv^* \Vert_F =1 $ and $ \Vert z \Vert= \nu $.
Next, we will estimate the left-hand side of (\ref{ineq:chain1}) by
\begin{equation}\label{ineq:chain3}
\begin{split}
&\big\|\Pi_{\widehat J_1}[\mathcal{A}^\ast(b)]\Pi_{ \left\{ k_2 \right\} }\big\|_\mathrm{F}\\  
\leq & \big\|\Pi_{ \widehat J_1 } uv^* \Pi_{ \left\{ k_2 \right\} }\big\|_\mathrm{F} +  \left(  \delta \Vert uv^* \Vert_F  +  \sqrt{1 + \delta}  \Vert z \Vert \right) \\
\le &    \big\|\Pi_{ \widehat J_1 } u   \big\|  \big\|  \Pi_{ \left\{ k_2 \right\} } v   \big\| + \left( \delta + \nu +\delta \nu  \right) \\
\le &    \big\|\Pi_{ \widehat J_1 } u   \big\|  \big\|  \Pi_{ \widehat J_2 } v   \big\| +   \left( \delta + \nu +\delta \nu  \right).
\end{split}
\end{equation}
The first two lines are obtained by an analogous reasoning as for (\ref{ineq:chain2}). The last line is due to $  \left\{ k_2 \right\}  \subset \widehat{J_2} $, which is a consequence of the definition of $ \widehat J_2 $ (\ref{equ:defwidetilde}) and the definition of $ \left\{ k_2 \right\} $ (\ref{equ:defwidehat}). We finish the proof by combining the inequality chains (\ref{ineq:chain1}), (\ref{ineq:chain2}), and (\ref{ineq:chain3}).
\end{proof}
%	%\end{proof}
%	
\end{document}